\documentclass{article}

\usepackage{biblatex}
\usepackage{hyperref}
\usepackage{comment}
\usepackage{color}
\usepackage{graphicx}
\usepackage{authblk}

\usepackage{physics}
\usepackage{amssymb}
\usepackage{amsmath}
\usepackage{amsthm}
\usepackage{algorithm}
\usepackage{algpseudocode}

\usepackage{tikz}
\usetikzlibrary{positioning}

\newtheorem{defn}{Definition}
\newtheorem{prob}{Problem}
\newtheorem{thrm}{Theorem}
\newtheorem{lemm}{Lemma}
\newtheorem{coro}{Corollary}
\newtheorem{remark}{Remark}

\title{A novel framework for Shot number minimization in Quantum Variational Algorithms}
\author[1,2]{Seyed Sajad Kahani \thanks{sajad@alqemie.co.uk}}
\author[1,3]{Amin Nobakhti  \thanks{amin@alqemie.co.uk}}
\affil[1]{Alqemie Ltd.}
\affil[2]{Department of Physics and Astronomy, University College London}
\affil[3]{Department of Electrical Engineering, Sharif University of Technology}
\date{July 8, 2023}

\begin{document}

\maketitle

\begin{abstract}
Variational Quantum Algorithms (VQAs) have gained significant attention as a potential solution for various quantum computing applications in the near term. However, implementing these algorithms on quantum devices often necessitates a substantial number of measurements, resulting in time-consuming and resource-intensive processes. This paper presents a generalized framework for optimization algorithms aiming to reduce the number of shot evaluations in VQAs. The proposed framework combines an estimator and an optimizer. We investigate two specific case studies within this framework. In the first case, we pair a sample mean estimator with a simulated annealing optimizer, while in the second case, we combine a recursive estimator with a gradient descent optimizer. In both instances, we demonstrate that our proposed approach yields notable performance enhancements compared to conventional methods.
\end{abstract}

\section{Introduction}

Variational Quantum Algorithms \cite{cerezo2021} have emerged as a promising solution for near-term applications of quantum computers. These versatile algorithms offer the capability to tackle a diverse range of complex problems, including but not limited to quantum chemistry \cite{peruzzo2014}, combinatorial optimization \cite{farhi2014}, and machine learning \cite{schuld2014}. Despite their potential for near-term applications, variational algorithms often require a large number of measurements. This makes implementation of those algorithms on quantum devices extremely time and resource-intensive \cite{harrigan2021, gaqac2020}, even when performed on shallow and low-width circuits.

Various research efforts have sought to employ optimizers to reduce the computational burden of VQAs. These include application of both existing and novel optimization techniques \cite{lavrijsen2020, pellowjarman2021, lockwood2022}. Such approaches are related to well studied and rich literature on optimization of noisy functions in various fields such as signal processing and control theory (see for example \cite{kelley2011} and \cite{chen2002}). Sweke et al.\cite{sweke2020} introduced a quantum stochastic gradient descent optimizer that relies on a gradient estimator with a limited number of shots. They proved that with some simplifying assumptions this approach will converge to the optimal values. However, the convergence rate is dependent on the error of the estimator. In another study, Polloreno et al.\cite{polloreno2022} studied the robustness of a double simulated annealing optimizer against inherent quantum noise, even when only a few shots are available and the noise is noticeable.

Another approach to solve this problem has been to employ a nested optimization framework in which a high-level optimizer is used to improve the performance of a low-level optimizer by tuning its parameters. For example, Tamiya et al.\cite{tamiya2022} employed Bayesian optimization on stochastic measurement results to determine the optimal step size through a line search. Inspired by stochastic gradient descent, this method incorporates an adaptive shot technique to reduce the number of measurements required during the line search. Similarly, Mueller et al.\cite{mueller2022} proposed a technique to identify a suitable initial value set using Gaussian Processes. Subsequently, they utilized ImFil as the optimizer in their approach.

In this work we propose a generalized framework for optimization algorithms which seek to reduce shot-number evaluations in VQAs. The key performance improving novelty in our approach are two fold.
First, devising a framework to incorporate powerful estimation techniques to achieve near-true parameter estimates with much fewer data samples.
Secondly, by utilizing the sensitivity analysis of the optimizers, it will be assured that the error level of estimators (and the number of shots as a result) are suitably chosen. This is made possible by breaking the problem into two separate estimation and optimization problems, and deriving theoretical results on the sufficient number of shot. We explore two specific case studies within this framework. For the first case, a sample mean estimator is paired with a simulated annealing optimizer, and in the second case, a recursive estimator is paired with a gradient descent optimizer.

The remainder of the paper is organized as follows; In section \ref{sec:concepts} background material, including quantum variational circuits, and estimation theory are presented. In section \ref{sec:discussion} we develop the proposed error control strategy and discuss the resulting optimization framework. In section \ref{sec:case_studies} we present two case studies together with numerical results. Finally, in section \ref{sec:conclusion}, we conclude our work.

\section{Basic Concepts}\label{sec:concepts}
\subsection{Quantum Variational Algorithms}

\def\C{\mathcal{C}}
\def\R{\mathbb{R}}

In theory of quantum variational algorithms, the expected value of an observable $O$ over a state, generated by applying the parameterized quantum circuit $U(\vb*{\theta})$ on the initial state $\ket{0}$ is a required data. This value is used by cost function $\C \in \R^m$ to be minimized with respect to the parameter space $\vb*{\theta}$. Accordingly, the class of algorithms such as VQE, QAOA and QNN, can be formulated as~\cite{cerezo2021},
\begin{equation}
  \label{eq:core_optimization}
  \vb*{\theta}^* = \min_{\vb*{\theta} \in \R^m} \C\qty( \ev{U(\vb*{\theta})^\dagger O U(\vb*{\theta})}{0} ).
\end{equation}

Specific details of these algorithms are available in \cite{cerezo2021}. Here we would like to focus on the underlying operation of these algorithms. Let,

\begin{equation}
    f^{U, O}(\vb*{\theta}) = \ev{U(\vb*{\theta})^\dagger O U(\vb*{\theta})}{0},
\end{equation}

in which $U$ and $O$ may be omitted when discussion is not related to the specific choice of $U$ and $O$. One of the simplest and widely used parameter-shift rules to compute the derivatives of $f$ is given in Lemma \ref{lemm:shift_rule}.

\def\ui{\mathrm{i}}
\begin{lemm}[Parameter-shift rule \cite{wierichs2022}]
\label{lemm:shift_rule}
under the circumstance that each the dependence of $f$ to each parameter (like $\vb*{\theta}_k$) is in the form of $e^{\ui \vb*{\theta}_k P_k}$ where $P_k$ is a Pauli operator, we have,
\begin{equation}
  \label{eq:shift_rule}
  \partial_k f(\vb*{\theta}) = \frac{f(\vb*{\theta} + \vu e_k \pi / 2) - f(\vb*{\theta} - \vu e_k \pi / 2)}{2}.
\end{equation}
\end{lemm}

Variable $\partial_k$ is $\pdv{\theta_k}$ and $\vu e_k$ is the vector with $1$ in the $k$-th position and $0$ elsewhere. Lemma \ref{lemm:shift_rule} is not only useful in calculating the derivative of $f$, it can also be used to bound higher derivatives of $f$ as shown in Lemma \ref{lemm:derivative_bound}.

\begin{lemm}
    \label{lemm:derivative_bound}
    For any $\vb*{\theta} \in \R^m$, we have,
    \begin{equation}
        \norm{\mathrm{Hess} f}_2 \le m\norm{O}_{2}.
    \end{equation}
\begin{proof}
    From the definition we know that $\abs{f} < \norm{O}_{2} \forall \vb*{\theta} \in \R^m$.
    For any $i$ and $j$ there always exist some values of $\vb*{\theta}_{1}, \vb*{\theta}_{2}, \vb*{\theta}_{3}, \vb*{\theta}_{4}$ for which,
    \begin{equation}
        \mathrm{Hess} f_{ij} = \frac{f(\vb*{\theta_{1}}) - f(\vb*{\theta}_{2}) - f(\vb*{\theta}_{3}) + f(\vb*{\theta}_{4})}{4} \le \norm{O}_{2}.
    \end{equation}
    Accordingly,
    \begin{equation}
        \norm{\mathrm{Hess} f}_2 \le m\norm{O}_{2}.
    \end{equation}
\end{proof}
\end{lemm}

\subsection{Estimation and Error Analysis}\label{ssec:estimation}
\def\Var{\mathrm{Var}}
\def\MSE{\mathrm{MSE}}
\def\Bias{\mathrm{Bias}}

Contrary to the simple definition of $f^{U, O}$, evaluating such an expected value at each sample point may involve measurements with respect to $\ell$ multiple bases. Accordingly, the observable $O$ will be decomposed to $\ell$ observables, each of which is diagonal in a different basis, such as,
\begin{equation}
    \label{eq:decomposition}
    O = \sum_{j=1}^\ell V^\dagger_{j} D_j V_j.
\end{equation}

For each $\ell$, it is necessary to perform $r_j$ repetitive measurements on a quantum circuit. The $l$th (out of $r_j$) measurement outcome will be considered as a sample from a random variable $\chi_{j, l} \sim X(UV_j, D_j, \vb*{\theta})$. We know that $\mathbb{E}[\chi_{j,l}] = f^{UV_j, D_j}(\vb*{\theta})$ and this is the reason we typically define an estimator $f^{U, O}(\vb*{\theta})$ as follows.

\begin{defn}[Sample Mean Estimator]
    \label{def:naive_estimator}
    A sample mean estimator for $f$ is defined as,
    \begin{equation}
        \hat{f}^{U, O}(\vb*{\theta}) = \sum_{j=1}^{\ell} \frac{1}{r_j} \sum_{l = 1}^{r_j} \chi_{j, l}.
    \end{equation}
    And for any of $\partial_k f$s,
    \begin{equation}
        \hat{\partial}_k f^{U, O}(\vb*{\theta}) = \sum_{j=1}^{\ell} \frac{1}{2r_{j+}} \sum_{l = 1}^{r_{j+}} \chi_{j+, l} - \frac{1}{2r_{j+}} \sum_{l = 1}^{r_{j-}} \chi_{j-, l}.
    \end{equation}
    where $\chi_{j+, l} \sim X(UV_j, D_j, \vb*{\theta} + \vu e_i \pi / 2)$ and $\chi_{j-, l} \sim X(UV_j, D_j, \vb*{\theta} - \vu e_i \pi / 2)$.
\end{defn}

The performance of such an estimator can be bounded with the aid of the Hoeffding's inequality. The inequality provides confidence intervals of the estimators of bounded random variables.

\def\E{\mathbb{E}}
\begin{lemm}[Hoeffding's inequality~\cite{hoeffding1963}]
    \label{lemm:hoeffding}
    For $n$ random variables $\xi_1, \xi_2, \dots, \xi_n$ with $a_i \le \xi_i \le b_i$ for all $i$, and any $t > 0$, we have,
    \begin{equation}
        \Pr\qty(\abs{\sum_{i=1}^n \xi_i - \sum_{i=1}^n \E[\xi_i]} \ge t) \le 2e^{\frac{-2t^2}{\sum_{i=1}^n (b_i - a_i)^2}}.
    \end{equation}
\end{lemm}

Based on this, the following bounds are obtained for the MSE (mean square error) and confidence interval (CI) of the sample mean estimator.

\begin{thrm}[Sample mean estimator bounds]
    \label{thrm:naive_bound}
    By defining,
    \begin{equation}
        \label{eq:shots_and_epsilon_f}
        \epsilon_f = \sum_{j=1}^{\ell} \frac{\norm{D_j}^2_2}{r_j},
    \end{equation}
    and,
    \begin{equation}
        \label{eq:shots_and_epsilon_df}
        \epsilon_{\partial_k f} = \sum_{j=1}^{\ell} \frac{\norm{D_j}^2_2}{4} \qty(\frac{1}{r_{j+}} + \frac{1}{r_{j-}}).
    \end{equation}
    When $\hat s$ is $\hat f^{U, O}$ or $\hat \partial_k f^{U, O}$, it can be respectively bounded by $\epsilon_f$ and $\epsilon_{\partial_k f}$ for any $\vb *{\theta}$ and $\kappa > 0$ as follows,
    \begin{equation}
        \MSE[ \hat{s}(\vb*{\theta})] \le \epsilon, \quad \Pr(\abs{\hat{s}(\vb*{\theta}) - s(\vb*{\theta})} > \kappa\sqrt{\epsilon}) \le 2e^{-\frac{\kappa^2}{2}}.
    \end{equation}
\end{thrm}
\begin{proof}
    To prove the bounds for $f$, we start by setting $\xi$s in Hoeffding's inequality to $\frac{\chi_{j,l}}{r_j}$ for different $j$ and $l$s. They are bounded to $-\frac{\norm{D_j}}{r_j} \le \frac{\chi_{j,l}}{r_j} \le \frac{\norm{D_j}}{r_j}$, it can thus be shown that,
    \begin{equation}
        \Pr(\abs{\hat{f}(\vb*{\theta}) - f(\vb*{\theta})} > t) \le 2e^{-\frac{2t^2}{4\epsilon_f}}.
    \end{equation}
     It is now only required to replace $t$ with $\kappa \sqrt{\epsilon_f}$. From Popoviciu's inequality~\cite{popoviciu1935} it is evident that $\Var[\xi_i] \le \frac{b_i - a_i}{4}$ which is used for the MSE of bounded random variables. The same results hold for the partial derivatives, if we set $\xi$s to $\frac{\chi_{j\pm,l}}{2r_{j\pm}}$ for different $j$ and $l$ and $+$ and $-$ signs.
\end{proof}

\section{Main Results}\label{sec:discussion}
\subsection{Error Control Strategy}

As mentioned in the introduction, a key performance improving novelty of our work is the means to control the error level, as well as the number of shots. This will be possible by connecting the number of shots to the error level of any estimator, using the problem below. Contrary to the normal estimators that often use a constant number of shots without any further analysis, we intend to find a sufficient value for $r_j$s such that the resulting estimation error is bounded by a specified amount.

\begin{prob}[Sufficient Number of Shots]
    \label{prob:sufficient_shots}
    Given an estimator $\hat s$, find the values of $r_j$s which satisfy the following constraints,
    \begin{equation}
        \label{eq:constraint_mse_f}
        \MSE[\hat{s}] \le E_{s}.
    \end{equation}
\end{prob}

For the sample mean estimator discussed previously, solving Problem \ref{prob:sufficient_shots}, for $f^{U, O}$ and $\partial_k f^{U, O}$ is equivalent to the following optimisation problems,

\begin{equation}
\label{eq:EBE_f}
    \underset{r_{j}\in \mathbb{N}}{\mathrm{argmin}}  \sum_{j=1}^{\ell} r_j \quad \mathrm{s.\ t.} \quad \MSE[\hat{f}] \le E_f.
\end{equation}
\begin{equation}
\label{eq:EBE_df}
    \underset{r_{j\pm}\in \mathbb{N}}{\mathrm{argmin}} \sum_{j=1}^{\ell} r_{j+} + r_{j-} \quad \mathrm{s.\ t.} \quad \MSE[\hat{\partial}_k f] \le E_{\partial_k f}.
\end{equation}

Optimization problems \ref{eq:EBE_f} and \ref{eq:EBE_df} can be approximately solved using Algorithm~\ref{algo:error_control}. This algorithm solves the optimisations by relaxing MSE values to the bounds $\epsilon_f$ and $\epsilon_{\partial_k f}$ defined in Theorem \ref{thrm:naive_bound} and limiting $r_j$s and $r_{j\pm}$s to have real values.

\begin{algorithm}
        a) Sufficient shots for $\hat{f}$, the function returns the outcoming bound for the error ($\epsilon_f$) as well as the number of shots ($r_j$s).
    \begin{algorithmic}
        \Function{ShotsForSMEstimatorF}{$E_f$}

        \State Decompose $f$ to $\ell$ terms (as \ref{eq:decomposition})
        \State $\nu \gets \qty(\sum_{j=1}^\ell \norm{D_j}_2)/E_f$
        \For {$j=1$ to $l$}
            \State $r_j \gets \lceil\qty(\norm{D_j}_2)/\nu\rceil$
        \EndFor
        \State $\epsilon_f \gets \sum_{j=1}^\ell \norm{D_j}^2_2/r_j$
        \State \Return ($r_j$s, $\epsilon_f$)
        \EndFunction
    \end{algorithmic}
        b) Sufficient shots for $\hat{\partial}_k f$ that returns the similar outputs.
    \begin{algorithmic}
        \Function{ShotsForSMEstimatorDF}{$E_{\partial_k f}$}
        \State $\nu \gets 2\qty(\sum_{j=1}^\ell \norm{D_j}_2)/E_f$
        \For {$j=1$ to $l$}
        \For {$\sigma$ in $\{+, -\}$}
            \State $r_{j\sigma} \gets \lceil\qty(\norm{D_j}_2)/\nu\rceil$
        \EndFor
        \EndFor
        \State $\epsilon_{\partial_k f} \gets \sum_{j=1}^\ell \norm{D_j}^2_2/4\qty(1/r_{j+} + 1/r_{j-})$
        \State \Return $\epsilon_{\partial_k f}$, $r_{j\pm}$s
        \EndFunction
    \end{algorithmic}
    \caption{Error control of sample mean estimators}
    \label{algo:error_control}
\end{algorithm}

We can easily verify the algorithm by replacing the values using the formulas in Theorem~\ref{thrm:naive_bound} and deduce that the algorithm not only bounds the MSE but also provides a CI for the values.

\subsection{Optimizing Agent}

Regardless of technical detail, the function of all variational algorithms can be considered as that of agent which interacts with a quantum computer as shown in Figure~\ref{fig:optimizer}. Such a high level conceptualization permits development of a unified framework for the evaluation of $f$, $\partial_k f$ and higher derivatives.

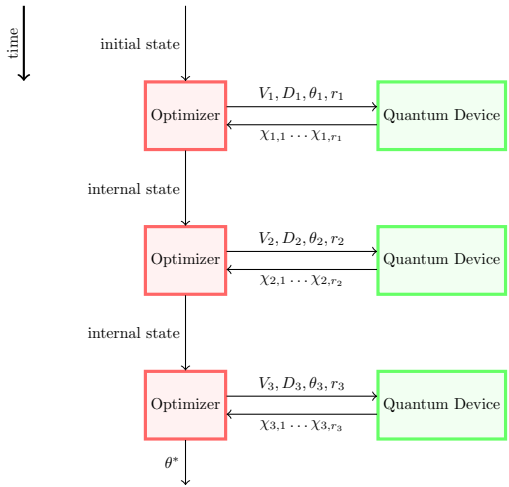
\begin{figure}[h]
    \centering
    \begin{tikzpicture}[
        scale=0.6, every node/.style={scale=0.6},
        align=center,node distance=1cm and 2cm,
        squared/.style={rectangle, minimum height=1.5cm, align=center},
        opt/.style={draw=red!60, fill=red!5, very thick},
        ctr/.style={draw=blue!60, fill=blue!5, very thick},
        qnt/.style={draw=green!60, fill=green!5, very thick},
        ]
        \node (o0) {};
        \node (t0) [left=of o0] {};
        \node (t1) [below=of t0] {};

        \draw[->, thick] (t0.south) -- (t1.north) node[midway,rotate=90,above] {time};
        \node[squared, opt] (o1) [below=of o0] {Optimizer};
        \node[squared, qnt] (q1) [right=2.0cm of o1] {Quantum Device};

        \node[squared, opt] (o2) [below=of o1] {Optimizer};
        \node[squared, qnt] (q2) [below=of q1] {Quantum Device};

        \node[squared, opt] (o3) [below=of o2] {Optimizer};
        \node[squared, qnt] (q3) [below=of q2] {Quantum Device};

        \draw[->] (o0.south) -- (o1.north) node[midway,left] {initial state};

        \draw[->] ([yshift=2mm]o1.east)  -- ([yshift=2mm]q1.west) node[midway,above] {$V_1, D_1, \theta_1, r_1$};
        \draw[->] ([yshift=-2mm]q1.west)  -- ([yshift=-2mm]o1.east) node[midway,below] {$\chi_{1, 1} \dots \chi_{1, r_1}$};

        \draw[->] (o1.south) -- (o2.north) node[midway,left] {internal state};

        \draw[->] ([yshift=2mm]o2.east)  -- ([yshift=2mm]q2.west) node[midway,above] {$V_2, D_2, \theta_2, r_2$};
        \draw[->] ([yshift=-2mm]q2.west)  -- ([yshift=-2mm]o2.east) node[midway,below] {$\chi_{2, 1} \dots \chi_{2, r_2}$};

        \draw[->] (o2.south) -- (o3.north) node[midway,left] {internal state};

        \draw[->] ([yshift=2mm]o3.east)  -- ([yshift=2mm]q3.west) node[midway,above] {$V_3, D_3, \theta_3, r_3$};
        \draw[->] ([yshift=-2mm]q3.west)  -- ([yshift=-2mm]o3.east) node[midway,below] {$\chi_{3, 1} \dots \chi_{3, r_3}$};

        \draw[->] (o3.south) -- ++(0,-1.0) node[midway,left] {$\theta^*$};
    \end{tikzpicture}
    \caption{The basic schematic of an optimizer where $\chi_{j, l} \sim X(V_j, D_j, \vb*{\theta_j})$}
    \label{fig:optimizer}
\end{figure}

Most general purpose optimizers will not aim to control the number of shots which is often taken as a constant during the optimization. There have been attempts to develop adaptive algorithms such as \cite{tamiya2022} but the scope of their application is limited. Any optimizing agent will ultimately utilize available data by calculating a set of estimators. Statistically, it is possible to reduce the number of estimators to a sufficient set of estimators. For most typical optimizer, those estimates will be limited to $\hat f^{U, O}(\theta_i)$ and $\hat{\partial}_k f^{U, O}(\theta_i)$, where $f^{U, O}$ is the function that is being optimized.

However, by application of sufficient shot problem proposed earlier, it is possible to control the optimization error, instead of the number of shots. In our view this is a more natural way of looking at the problem. In such an improved strategy, the optimizer is provided with the errors $E_{f}$ and $E_{\partial_k f}$ instead of $r_j$, and solves for $\hat{f}$, $\hat{\partial}_k f$ instead of $\chi_{j, l}$. This is illustrated in Figure~\ref{fig:interactive_optimizer_with_estimators}.

For the sake of simplicity we shall henceforth refer to $f^{U, O}(\theta_i)$ and $\partial_k f^{U, O}(\theta_i)$ as $f_i$ and $\partial_k f_i$ respectively. Moreover, this strategy can also be extended to the sample mean estimator $\hat{f}_i$ and $\hat{\partial_k}f_i$, defined in Definition~\ref{def:naive_estimator}.

\begin{figure}[h]
    \centering
    \begin{tikzpicture}[
        scale=0.6, every node/.style={scale=0.6},
        align=center,node distance=1cm and 2cm,
        squared/.style={rectangle, minimum height=1.5cm, align=center},
        opt/.style={draw=red!60, fill=red!5, very thick},
        ctr/.style={draw=blue!60, fill=blue!5, very thick},
        qnt/.style={draw=green!60, fill=green!5, very thick},
        ]
        \node (o0) {};
        \node (t0) [left=of o0] {};
        \node (t1) [below=of t0] {};

        \draw[->, thick] (t0.south) -- (t1.north) node[midway,rotate=90,above] {time};

        \node[squared, opt] (o1) [below=of o0] {Optimizer w. \\ sensitivity analysis};
        \node[squared, ctr] (c1) [right=1.4cm of o1] {Estimator w. \\ sufficient shots};
        \node[squared, qnt] (q1) [right=2.0cm of c1] {Quantum Device};

        \node[squared, opt] (o2) [below=of o1] {Optimizer w. \\ sensitivity analysis};
        \node[squared, ctr] (c2) [below=of c1] {Estimator w. \\ sufficient shots};
        \node[squared, qnt] (q2) [below=of q1] {Quantum Device};

        \node[squared, opt] (o3) [below=of o2] {Optimizer w. \\ sensitivity analysis};
        \node[squared, ctr] (c3) [below=of c2] {Estimator w. \\ sufficient shots};
        \node[squared, qnt] (q3) [below=of q2] {Quantum Device};
        \draw[->] (o0.south) -- (o1.north) node[midway,left] {inital state};

        \draw[->] ([yshift=4mm]o1.east)  -- ([yshift=4mm]c1.west) node[midway,above] {$\theta_1$, $E_{f_1}$, $E_{\partial_k, f_1}$};
        \draw[->] ([yshift=-4mm]c1.west)  -- ([yshift=-4mm]o1.east) node[midway,below] {$\hat{f}(\theta_1)$, $\hat{\grad}f(\theta_1)$};
        \draw[->] ([yshift=2mm]c1.east) -- ([yshift=2mm]q1.west) node[midway,above] {$r_j$s, $r_{j\pm}$s};
        \draw[->] ([yshift=-2mm]q1.west) -- ([yshift=-2mm]c1.east) node[midway,below] {$\chi_{j,l}$s, $\chi_{j\pm, l}$s};

        \draw[->] (o1.south) -- (o2.north) node[midway,left] {internal state};
        \draw[->, thin, gray] (c1.south) -- (c2.north) node[midway,left,align=right,gray] {internal state};
        \draw[->] ([yshift=2mm]o2.east)  -- ([yshift=2mm]c2.west) node[midway,above] {$\theta_2$};
        \draw[->] ([yshift=-2mm]c2.west)  -- ([yshift=-2mm]o2.east) node[midway,below] {$\hat{f}(\theta_2)$, $\hat{\grad}f(\theta_2)$};
        \draw[->] ([yshift=2mm]c2.east) -- ([yshift=2mm]q2.west) node[midway,above] {$r_j$s, $r_{j\pm}$s};
        \draw[->] ([yshift=-2mm]q2.west) -- ([yshift=-2mm]c2.east) node[midway,below] {$\chi_{j,l}$s, $\chi_{j\pm, l}$s};

        \draw[->] (o2.south) -- (o3.north) node[midway,left] {internal state};
        \draw[->, thin, gray] (c2.south) -- (c3.north) node[midway,left,align=right,gray] {internal state};
        \draw[->] ([yshift=2mm]o3.east)  -- ([yshift=2mm]c3.west) node[midway,above] {$\theta_3$};
        \draw[->] ([yshift=-2mm]c3.west)  -- ([yshift=-2mm]o3.east) node[midway,below] {$\hat{f}(\theta_3)$, $\hat{\grad}f(\theta_3)$};
        \draw[->] ([yshift=2mm]c3.east) -- ([yshift=2mm]q3.west) node[midway,above] {$r_j$s, $r_{j\pm}$s};
        \draw[->] ([yshift=-2mm]q3.west) -- ([yshift=-2mm]c3.east) node[midway,below] {$\chi_{j,l}$s, $\chi_{j\pm, l}$s};

        \draw[->] (o3.south) -- ++(0,-1.0) node[midway,left] {$\theta^*$};

    \end{tikzpicture}
    \caption{Schematic diagram of an optimizer with sensitivity analysis and an estimator with a sufficient shot algorithm.}
    \label{fig:interactive_optimizer_with_estimators}
\end{figure}
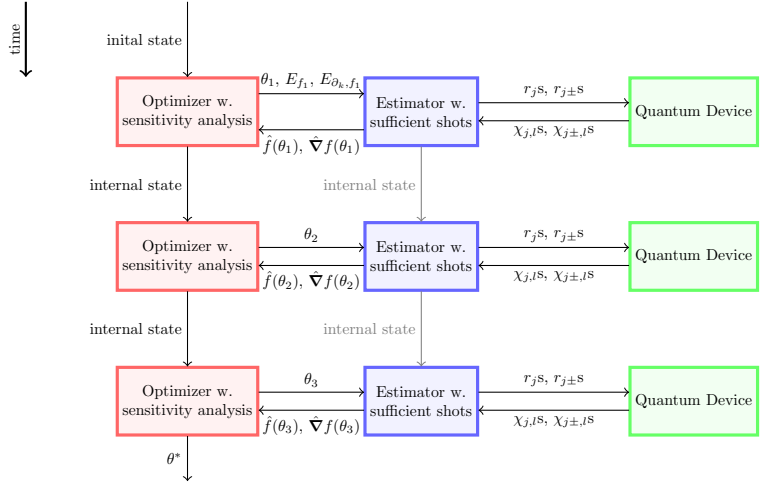

In the proposed framework the main problem is broken down into two separate problems. These are,
\begin{enumerate}
    \item An optimization problem of uncertain values, with a sensitivity analysis
    \item An estimation problem, with the question of sufficient shots for the estimator.
\end{enumerate}

In the proposed framework one is not limited to the sample mean estimator defined in Definition \ref{def:naive_estimator} and can make use of any static or dynamic estimator. Dynamic estimators will also have an internal states which is shown by a gray arrow in Figure \ref{fig:interactive_optimizer_with_estimators}.

We will demonstrate the profound effectiveness of this approach by introducing a few examples of estimators and optimizers in the following section. For the sake of illustrating the methodology we shall make use of existing standard and rather simple optimization and estimation techniques. Evidently the eventual obtainable performance improvements can be much greater by a well matched and individually powerful optimizer and estimator.

\section{Case Studies}\label{sec:case_studies}
\subsection{Example I: Error-Aware Simulated Annealing}
A simple simulated annealing algorithm is a stochastic process that starts from a random point in the search space and iteratively moves to a new point with a transition probability $P$ based on the values and temperature $T_i$ at step $i$. In order to introduce the uncertainty, we only need to redefine the transition probability $\hat P$ based on the estimator as follows,

\begin{equation}
    \hat P(\vb*{\theta}_{i+1} | \vb*{\theta}_i) = \begin{cases}
        1 & \text{if } \hat{f}_{i+1} < \hat{f}_i \\
        e^{-\frac{\hat{f}_{i+1} - \hat{f}_i}{T_i}} & \text{otherwise}.
    \end{cases}
\end{equation}

Then, the sensitivity can be analyzed as follows. In order to maintain an accuracy for $\hat P(\vb*{\theta}_{i+1} | \vb*{\theta}_i)$ we seek,

\begin{equation}
    \E\qty[D_{KL}(P \parallel \hat P)] \le \eta,
\end{equation}

where $D_{KL}$ is the Kullback-Leibler divergence. We know that this equation will hold if,

\begin{equation}
    \E\qty[\abs{\log \frac{P(\vb*{\theta}_{i+1} | \vb*{\theta}_i)}{\hat P(\vb*{\theta}_{i+1} | \vb*{\theta}_i)}}] \le \eta \qquad \forall \vb*{\theta}_{i+1}.
\end{equation}

The RHS could be bounded using $\E[\abs{x - \E[x]}] \le \sqrt{\Var[x]}$ and the independence of $\hat{f}_{i+1}$ and $\hat{f}_i$ and by assuming a monotonically decreasing temperature $T_{i+1} < T_{i}$,
\begin{equation}
    \begin{aligned}
    \E\qty[\abs{ \log P(\vb*{\theta}_{i+1} | \vb*{\theta}_i) - \log \hat P(\vb*{\theta}_{i+1} | \vb*{\theta}_i)}] &\le \frac{1}{T_i}\E\qty[\abs{ \hat{f}_{i+1} - \hat{f}_i - f_{i+1} + f_i}],  \\
    &\le \frac{1}{T_i}\sqrt{\Var\qty[\hat{f}_{i+1} - \hat{f}_{i}]}, \\
    &\le \frac{1}{T_i}\sqrt{\Var\qty[\hat{f}_{i+1}] + \Var\qty[\hat{f}_{i}]} .
    \end{aligned}
\end{equation}

Note that the estimators should be unbiased, otherwise the equation above will not hold. Finally we will introduce the condition below, that is sufficient for the equation above and furthermore to bound KL divergence by $\eta$,
\begin{equation}
    \label{eq:intellig_condition}
    \MSE[f_{i+1}] \le \frac{\eta^2 T^2_i}{2}.
\end{equation}

This is a more efficient condition for the estimator in comparison to the simply asking $\MSE[f_{i+1}] \le E$. In order to compare the performance of the simulated annealing with and without the sensitivity analysis, we conducted three experiments as follows,
\begin{itemize}
    \item \textbf{Simple Optimizer (1)}: A simulated annealing optimizer with the condition $\MSE[f_{i+1}] \le E$ with a high value for $E$.
    \item \textbf{Simple Optimizer (2)}: A simulated annealing optimizer with the condition $\MSE[f_{i+1}] \le E$ with a low value for $E$.
    \item \textbf{Error-Aware Optimizer}: A simulated annealing optimizer with Equation \ref{eq:intellig_condition} as the condition.
\end{itemize}

For experimental studies, consider the benchmark problem defined in \ref{prob:toy}.

\begin{prob}[Benchmark problem]\label{prob:toy}
    Assume a variational task with one qubit and $U(\theta) = R_x(\theta)$ and $O = Z$ with $\mathcal{C} = I$, which implies $\ell = 1$ and $m = 1$. Also $C(\theta) = \ev{R_x^\dagger(\theta) Z R_x^\dagger(\theta)}{0}$ could be simplified further into $\cos\theta$.
\end{prob}

We start with an ensemble of $\theta$s near $0$ and compare the distribution of the exact value of the function $f$ through the optimization (with respect to the number of shots conducted) for each optimizer. The results are shown in Figure \ref{fig:sa}.

\begin{figure}[h]
    \centering
    \includegraphics[width=0.8\textwidth]{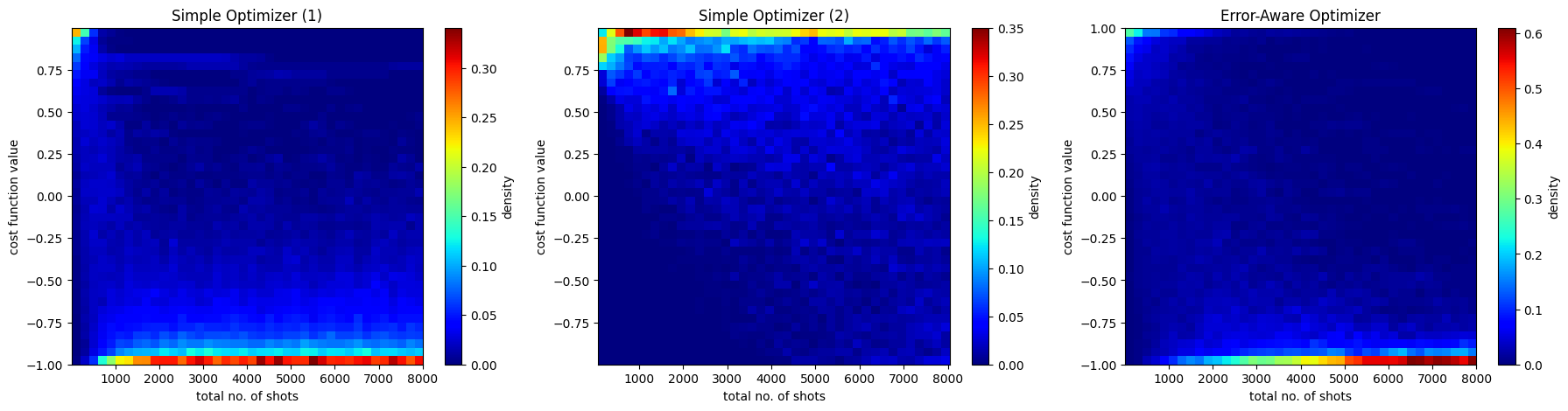}
    \caption{Comparison of the performance of the error-aware simulated annealing with the simpler ones.}
    \label{fig:sa}
\end{figure}

To more clearly highlight the difference between the distributions, we have also plotted the distribution of data points after $7000$ shots for each optimizer in Figure \ref{fig:sa_dist}.

\begin{figure}[h]
    \centering
    \includegraphics[width=0.5\textwidth]{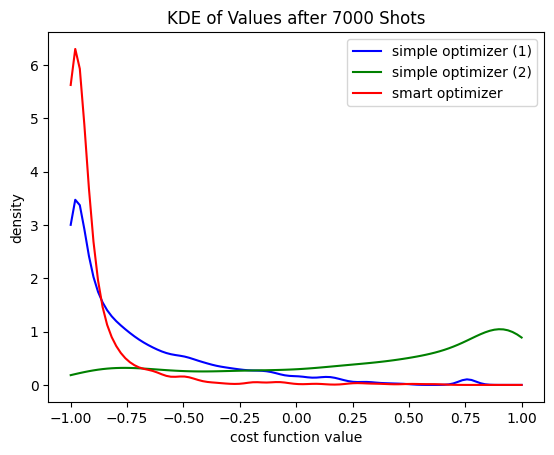}
    \caption{Distribution of datapoints after $7000$ shots for each optimizer.}
    \label{fig:sa_dist}
\end{figure}

Note that the error bound for different optimizers as a function of the number of shots is shown in Figure \ref{fig:sa_error} which is just a visualisation of  condition \ref{eq:intellig_condition}.

\begin{figure}[h]
    \centering
    \includegraphics[width=0.5\textwidth]{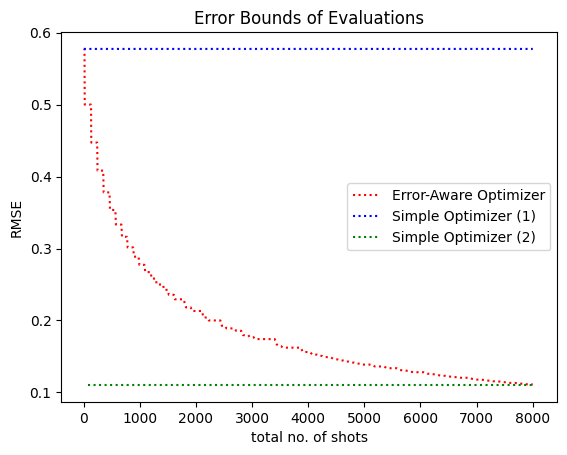}
    \caption{Error bound for different optimizers as a function of the number of shots.}
    \label{fig:sa_error}
\end{figure}

The results show that the error-aware simulated annealing is able to find a better solution with less number of shots.

\subsection{Example II: Recursive Estimator for Gradient Descent}

To illustrate the flexibility of the framework with respect to the choice of estimators and optimizers, in this section we perform experiments with a standard gradient descent algorithm and a novel recursive estimator for the function and its derivative. The proposed recursive estimator works on the assumption that the distance between two function evaluations required by the optimizer at two consecutive iterations is not great. That is, the function (and possibly its gradient) at a point $\vb*{\theta}_i$ and its next evaluation at $\vb*{\theta}_{i+1}$ doesn't differ drastically from $\vb*{\theta}_i$. This assumption allows the update rule of the optimizer to be written in the form $\vb*{\theta}_{i+1} = \vb*{\theta}_i + \delta \vb*{\theta}_i$ where $\delta \vb*{\theta}_i$ is a vector with bounded norm. The proposed recursive estimation methodology is formally defined in Definition \ref{def:sofr_estimator}.

\begin{defn}[Recursive Estimators]
    \label{def:sofr_estimator}
    \begin{equation}
        \begin{cases}
        \hat{f}^{*}_{i} = \alpha_i(\hat{f}^{*}_{i-1} + \delta\vb*{\theta}_{i-1} \cdot \hat{\grad}f^{*}_{i-1}) + (1 - \alpha_i) \hat{f}_{i} \\
        \hat{\partial_k}f^{*}_{i} = \beta_i \hat{\partial_k}f^{*}_{i-1} + (1 - \beta_i) \hat{\partial_k}f_{i}
        \end{cases}
        , \quad
        \begin{cases}
            \hat{f}^{*}_0 = \hat{f}_0 \\
            \hat{\partial_k}f^{*}_0 = \hat{\partial_k}f_0
        \end{cases}
    \end{equation}
\end{defn}

Note that $\alpha_i$s and $\beta_i$s are values between $0$ and $1$ and act as hyperparameters which control the relative weight given to prior knowledge. The optimal values of these parameters are derives in later sections. First we present Theorem \ref{thrm:sofr_bound} which derives theoretical bounds for the bias and variance of the estimate so obtained.

\begin{thrm}[Recursive estimator bounds]\label{thrm:sofr_bound}
    For any $i$,
    \begin{equation}
        \begin{cases}
            \Bias[\hat{f}^{*}_{i}] \le B_i \\
            \Bias[\hat{\partial_k} f^{*}_{i}] \le B_{\partial_k, i}.
        \end{cases}
    \end{equation}
    Where $B_i$ and $B_{\partial_k, i}$ are calculated recursively as follows,
    \begin{equation}
        \begin{cases}
            B_i = \alpha_{i}\qty(B_{i-1} + \sum_{k=1}^m \abs{\qty(\delta\vb*{\theta}_{i-1})_k} B_{\partial_k, i-1} + \frac{m}{2}\norm{\delta\vb*{\theta}_{i-1}}_2^2 \norm{O}_2) \\
            B_{\partial_k, i} = \beta_{k,i} \qty(B_{\partial_k, i-1} + \norm{\delta\vb*{\theta}_{i-1}}_2 \norm{O}_2)
        \end{cases}, \quad
        \begin{cases}
            B_0 = 0 \\
            B_{\partial_k, 0} = 0.
        \end{cases}
    \end{equation}
    and similarly for the variance,
    \begin{equation}
        \begin{cases}
        \Var[\hat{f}^*_i] \le A^2_i \\
        \Var[\hat{\partial_k}f^*_i] \le A^2_{\partial_k, i}.
        \end{cases}
    \end{equation}
    Using the notation in, Theorem~\ref{thrm:naive_bound}
    \begin{equation}
        \begin{cases}
        A^2_i = \alpha_i^2 \qty(A^2_{i-1} + \sum_{k=1}^m \abs{\qty(\delta\vb*{\theta}_{i-1})_k}^2 A^2_{\partial_k, i-1})  + \qty(1 - \alpha_i)^2 \epsilon^2_{f_i} \\
        A^2_{\partial_k, i} = \beta_{k,i}^2 A^2_{\partial_k, i-1} + \qty(1 - \beta_{k,i})^2 \epsilon^2_{\partial_k f_i},
        \end{cases}
    \end{equation}
\end{thrm}
\begin{proof}

Defining the drift term $d_i = f_{i - 1} + \delta\vb*{\theta}_{i-1} \cdot \grad f_{i-1} - f_{i}$, we can write the bias and variance of $\hat{f}^*_i$ as,
\begin{align}
    \Bias\qty[\hat{f}^*_i] &= \alpha_i \qty(\Bias\qty[\hat{f}^*_{i-1}] + \delta\vb*{\theta}_{i-1} \cdot \Bias\qty[\hat{\grad}f^*_{i-1}] + d_i) \\
    \Var\qty[\hat{f}^*_i] &= \alpha_i^2 \qty(\Var\qty[\hat{f}^*_{i-1}] + \delta\vb*{\theta}_{i - 1}^2\cdot\Var\qty[\hat{\grad}f^*_{i-1}]) + \qty(1 - \alpha_i)^2 \Var\qty[\hat{f}_i].
\end{align}
In an abuse of notation, $\delta\vb*{\theta}^2_{i-1}$ represents a vector of squared elements and $\Var\qty[\hat{\grad}f^*_{i-1}]$ represents a vector of variances. This facilitates a more compact proof as shall be seen. With the same objective, we define another drift term for the derivatives of $f$ as $d_{\partial_k, i} = \partial_k f_{i - 1} - \partial_k f_{i}$ will helps us to write the bias and variance of $\hat{\partial_k}f^*_i$ as,
\begin{align}
    \Bias\qty[\hat{\partial_k}f^*_i] &= \beta_{k,i} \qty(\Bias\qty[\hat{\partial_k}f^*_{i-1}] + d_{\partial_k, i}) \\
    \Var\qty[\hat{\partial_k}f^*_i] &= \beta_{k,i}^2 \Var\qty[\hat{\partial_k}f^*_{i-1}] + \qty(1 - \beta_{k,i})^2 \Var\qty[\hat{\partial_k}f_i].
\end{align}

Combining Lemma~\ref{lemm:derivative_bound} with the mean value theorem, we have,
\begin{equation}
\begin{cases}
    \abs{d_i} \le \frac{1}{2} \norm{\delta \vb*{\theta}_{i-1}}_2^2 m \norm{O}_2 \\
    \abs{d_{\partial_k, i}} \le \norm{\delta \vb*{\theta}_{i-1}}_2 \norm{O}_2.
\end{cases}
\end{equation}

Finally, combining the above equations with the fact that $\Var[\hat{f}_i] \le \epsilon^2_{f_i}$ and $\Var[\hat{\partial_k} f_i] \le \epsilon^2_{\partial_k f_i}$ completes the proof.
\end{proof}

For the confidence interval of recursive estimator, we can prove the following result,

\begin{coro}[Confidence Interval]
    As a result of Theorem~\ref{thrm:sofr_bound} the following equation is valid for $s^*$ is any of $f_i$s or $\partial_k f_i$s, simply by setting corresponding $A$ and $B$s.
    \begin{equation}
        \MSE[\hat{s}^*] \le B^2 + A^2, \quad \Pr(\abs{\hat{s}^* - f} > \kappa A + B) \le 2e^{-\frac{\kappa^2}{2}}.
    \end{equation}
\end{coro}
\begin{proof}

While the expression for the MSE is trivial, for the confidence interval we have,
\begin{equation}
    \Pr(\abs{\hat{f}^*_i - \E[\hat{f}^*_i]} > \kappa\sqrt{A_i}) \le 2e^{-\frac{\kappa^2}{2}}.
\end{equation}
This is true because $\hat{f}^*_i$ is a linear combination of $\chi$s that are from bounded distributions. Accordingly, Hoeffding's inequality applies. Moreover, there is a one-to-one correspondence between bounds from Hoeffding's and Popoviciu's inequalities (see the proof of Theorem~\ref{thrm:naive_bound}), which obviously validates the equation above. Since $\abs{\hat{f}^*_i - f_i} > \kappa\sqrt{A_i} + B_i \Rightarrow \abs{\hat{f}^*_i - \E[\hat{f}^*_i]} > \kappa\sqrt{A_i}$,
\begin{equation}
    \Pr(\abs{\hat{f}^*_i - f_i} > \kappa\sqrt{A_i} + B_i) \le \Pr(\abs{\hat{f}^*_i - \E[\hat{f}^*_i]} > \kappa\sqrt{A_i}) \le 2e^{-\frac{\kappa^2}{2}}.
\end{equation}
\end{proof}

Finally, we need to solve the sufficient shots problem (Problem \ref{prob:sufficient_shots}) for the recursive estimator. The actual objective is to solve,
\begin{equation}
    \begin{aligned}
    \underset{r_{j, i}, r_{j\pm,i} \in \mathbb{N}, \alpha_i, \beta_{k,i}}{\mathrm{argmin}} &\sum_{i=1}^\infty \sum_{j=1}^\ell r_{j, i} + \sum_{k=1}^m r_{j+, k, i} + r_{j-, k, i} \quad \\
    &\mathrm{s.\ t.} \quad \forall i \quad \MSE[\hat{f}^*_i] \le E_f  \\
    &\mathrm{s.\ t.} \quad \forall i, k \quad  \MSE[\hat{\partial}_k f^*_i] \le E_{\partial_k f}.
    \end{aligned}
\end{equation}

However, we solve an iterative version as in Algorithm \ref{algo:error_control},
\begin{equation}
    \min_{r_j \in \mathbb{N}, \alpha_i} \sum_{j=1}^\ell r_j \quad \mathrm{s.\ t.} \quad \MSE[\hat{f}^*_i] \le E_f.
\end{equation}
\begin{equation}
    \min_{r_{j,\pm} \in \mathbb{N}, \beta_{k,i}} \sum_{j=1}^\ell r_{j+} + r_{j-} \quad \mathrm{s.\ t.} \quad \MSE[\hat{\partial}_k f^*_i] \le E_{\partial_k f}.
\end{equation}

Combining the two leads to Algorithm \ref{algo:error_control_recursive}.

\begin{algorithm}
        a) Sufficient shots for $\hat{f}^*_i$
    \begin{algorithmic}
        \Function{ShotsForREstimatorF}{$E_f$}
        \State using $A_{i-1}$ and $B_{i-1}$ from the previous evaluations
        \State using $B_{\partial_k,i-1}$s and $A_{\partial_k,i-1}$s from \textsc{ShotsForREstimatorDF}

        \State $b \gets B_{i-1} + \norm{\delta\vb*{\theta}_{i-1}}_2 B_{\partial, i-1} + \frac{m}{2}\norm{\delta\vb*{\theta}_{i-1}}_2^2 \norm{O}_2$
        \State $a \gets A_{i-1} + \norm{\delta\vb*{\theta}_{i-1}}_2^2 A_{\partial, i-1}$

        \If {$b^2 + a > E_f$}
        \State $E' \gets \qty(b^2 + a) E_f/\qty(b^2 + a - E_f)$
        \State ($r_j$s, $\epsilon$) $\gets$ \Call{ShotsForSMEstimatorF}{$E'$}
        \Else
            \State $\epsilon \gets \infty$
            \For {$j=1$ to $l$}
            \State $r_j \gets 0$
        \EndFor
        \EndIf
        \State $\alpha_i \gets \qty(b^2 + a)/\qty(b^2 + a - E_f)$
        \State $A_i \gets \alpha_i^2 a + (1-\alpha_i)^2 \epsilon$
        \State $B_i \gets \alpha_i b$
        \State \Return ($r_j$s, $B_i^2 + A_i$)
        \EndFunction
    \end{algorithmic}
        b) Sufficient shots for $\hat{\partial}_k^*f_i$
    \begin{algorithmic}
        \Function{ShotsForREstimatorDF}{$E_{\partial_k f}$}
        \State using $A_{\partial_k,i-1}$ and $B_{\partial_k, i-1}$ from the previous evaluations

        \State $b \gets B_{\partial_k, i-1} + m\norm{\delta\vb*{\theta}_{i-1}}_2 \norm{O}_2$
        \State $a \gets A_{\partial_k, i-1} $

        \If {$b^2 + a > E_{\partial_k f}$}
            \State $E' \gets \qty(b^2 + a) E_{\partial_k f}/\qty(b^2 + a - E_{\partial_k f})$
            \State ($r_{j\pm}$s, $\epsilon$) $\gets$ \Call{ShotsForSMEstimatorDF}{$E'$}
        \Else
            \State $\epsilon \gets \infty$
            \For {$j=1$ to $l$}
            \For {$\sigma$ in $\{+, -\}$}
            \State $r_{j\sigma} \gets 0$
            \EndFor
        \EndFor
        \EndIf

        \State $\beta_{k,i} \gets \qty(b^2 + a)/\qty(b^2 + a - E_f)$
        \State $A_{\partial_k,i} \gets \beta_{k,i}^2 a + \beta_{k,i}^2 \epsilon$
        \State $B_{\partial_k,i} \gets \beta_{k,i} b$
        \State \Return ($r_{j\sigma}$s, $B_{\partial_k,i}^2 + A_{\partial_k,i}$)
        \EndFunction
    \end{algorithmic}
    \caption{Error control of recursive estimators}
    \label{algo:error_control_recursive}
\end{algorithm}

\begin{remark}
  Note that with this algorithm, for the same error bound, the number of shots for a recursive estimator of a function will be at max equal to the number of shots for the naive estimator of that function.
\end{remark}

To illustrate the performance of Algorithm \ref{algo:error_control_recursive}, first we apply the estimator for the variational Problem \ref{prob:toy} with a random (zero mean) initial point and a simple gradient-descent optimizer. Figure~\ref{fig:simple_compare} shows the estimated values (with CIs) of the loss function, for different estimators, as a function of the number of shots used to evaluate the function.

\begin{figure}[h]
    \label{fig:simple_compare}
    \centering
    \includegraphics[width=0.5\textwidth]{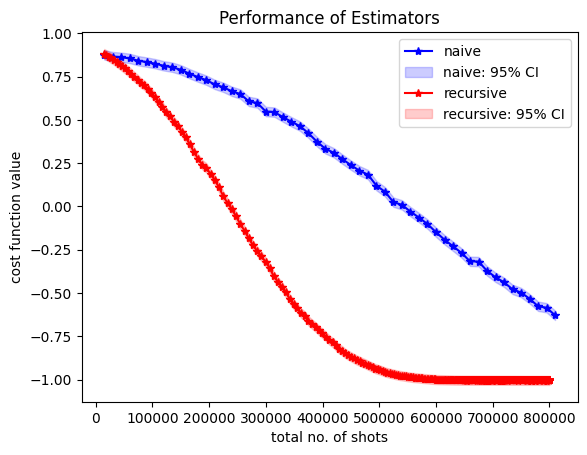}
    \caption{estimated loss value vs. number of shots, for a simple GD optimizer equipped with each estimator}
\end{figure}

It is evident that the proposed recursive estimator is outperforming the sample mean estimator by a significant margin. Another comparison made by visualizing number of shots per each GD iteration is shown in Figure~\ref{fig:agility_compare}.

\begin{figure}[h]
    \centering
    \includegraphics[width=0.5\textwidth]{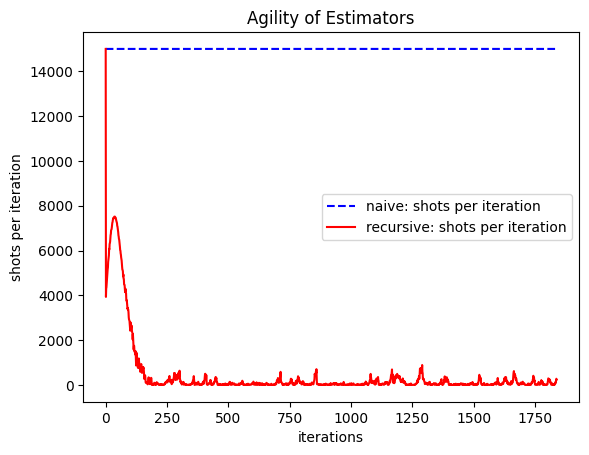}
    \caption{number of shots per each GD iteration for each of those estimators}
    \label{fig:agility_compare}
\end{figure}

To verify the theoretical results derived earlier, the bounds on MSE and CI are compared with the actual values of the MSE and CI of the estimators in Figures \ref{fig:mse_compare} and \ref{fig:ci_compare} respectively.

\begin{figure}[h]
    \centering
    \includegraphics[width=0.5\textwidth]{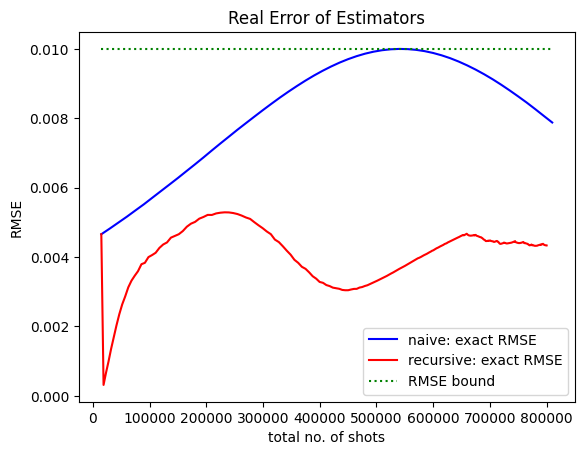}
    \caption{Exact MSE values vs Bounded MSE values}
    \label{fig:mse_compare}
\end{figure}

\begin{figure}[h]
    \centering
    \includegraphics[width=0.5\textwidth]{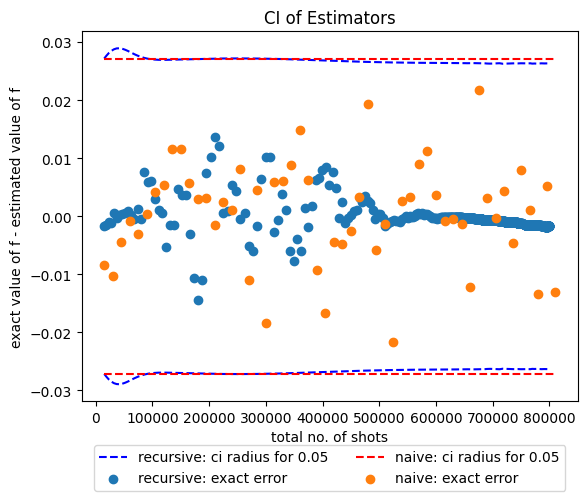}
    \caption{CI bounds and the difference between exact value and the estimated value of the function}
    \label{fig:ci_compare}
\end{figure}

For further experimental verification, the same experiment has also been carried out on the more complex MaxCut problem for a square graph ($\abs{V} = 4$ and $\abs{E} = 4$). The results are shown in Figure~\ref{fig:maxcut_compare} and Figure~\ref{fig:maxcut_agility_compare}.

\begin{figure}[h]
    \centering
    \includegraphics[width=0.5\textwidth]{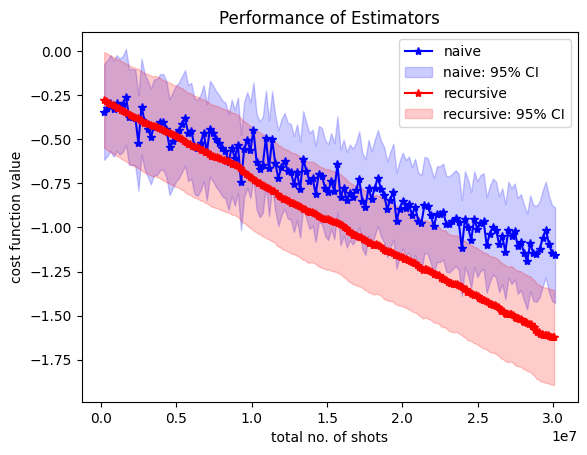}
    \caption{loss function vs. number of shots, for a simple GD optimizer equipped with each estimator}
    \label{fig:maxcut_compare}
\end{figure}

\begin{figure}[h]
    \centering
    \includegraphics[width=0.5\textwidth]{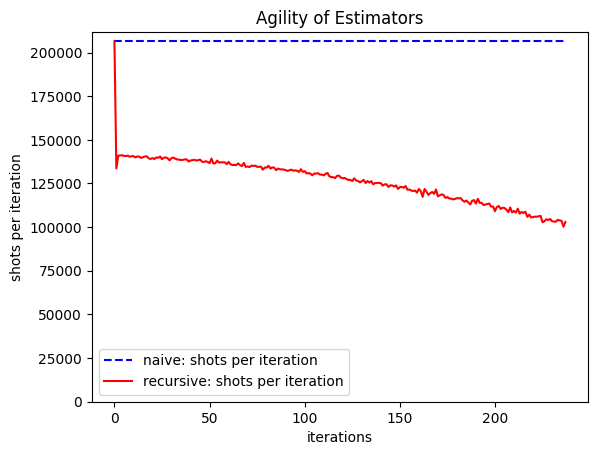}
    \caption{number of shots per each GD iteration for each of those estimators}
    \label{fig:maxcut_agility_compare}
\end{figure}

\section{Concluding remarks}\label{sec:conclusion}

In this paper, a generalized framework for optimization algorithms which seek to reduce shot-number evaluations in VQAs was proposed. In the general form, the proposed framework entails a combination of an estimator together with a numerical optimization algorithm. We introduced the sufficient shots problem and proposed an algorithm for it to be used with the sample mean estimator. This concept together with sensitivity analysis of optimizers, allows us to control the number of shots leading to a more natural and effective optimization process.

Two specific case studies of this framework were subject to extensive experiments. In the first case, a sample mean estimator is coupled with a simulated annealing optimizer, and in the second case, a recursive estimator was coupled with a gradient descent optimizer. In both cases we demonstrated that the proposed approach achieves significant performance improvements over conventional methods.

Our results highlight the importance of considering error control strategies and incorporating them into the design of optimizers for variational quantum algorithms. By leveraging estimators with error control and integrating them with interactive optimization processes, we can achieve better optimization performance and reduce the resource requirements for quantum computations.

Overall, this work contributes to advancing the field of variational quantum algorithms by providing a systematic framework for designing error-aware optimizers. The presented approaches and results open up new possibilities for improving the efficiency and effectiveness of quantum computing research in various domains, such as quantum chemistry, combinatorial optimization, and machine learning. Future directions could explore further extensions and applications of the proposed framework, as well as experimental validations on quantum devices.

\printbibliography

\section*{Appendix}

\begin{table}[!ht]
    \begin{tabular}{p{0.13\linewidth} | p{0.84\linewidth}}
        \textbf{Symbol} & \textbf{Description} \\
        \hline\hline
        $i$ & Iterations of the optimizer (except for Lemma \ref{lemm:hoeffding}) \\
        $j$ & Index of terms of the observable $O$ (from $1$ to $\ell$)\\
        $k$ & Index of dimensions of the parameter space (from $1$ to $m$) \\
        $l$ & Index of the shots\\
        $\sigma$ & Index for $+$ and $-$ \\
        $m$ & Dimension of the parameter space \\
        $\ell$ & Number of the terms in the decomposition of observable $O$ \\
        $r_j$, $r_{j, i}$, $r_{j \sigma, i}$, $r_{i\sigma, i}$ & Number of the shots \\
        \hline
        $\ket{0}$ & Initial state of the quantum circuit \\
        $\vb*{\theta}$ & Parameter vector \\
        $U(\theta)$ & Parameterized quantum circuit \\
        $O$ & Observable \\
        $\mathcal{C}$ & Cost function \\
        $f$ & The expectation value of the measurement \\
        $\vu{e}_k$ & Unit vector in the $k$-th dimension of parameter space \\
        $V_j$ & The gate for the $j$-th term of the observable $O$ \\
        $D_j$ & The measurement operator for the $j$-th term of the observable $O$ \\
        $\chi_{j, l}, \chi_{j \sigma, l}$ & The outcome of the measurement for the value or gradient terms \\
        $X($U$, $O$, \theta)$ & The distribution of the outcome \\
        $\xi_i$,$t$ & Only used in Hoefding's inequality (Lemma \ref{lemm:hoeffding}) \\
        $\epsilon_f$, $\epsilon_{\partial_k f}$ & The error bound for the estimators, defined in Theorem \ref{thrm:naive_bound} \\
        $s$ & a dummy function defined and used in Theorem \ref{thrm:naive_bound} \\
        $\kappa$ & A degree of freedom in the confidence interval \\
        $E_f$ & The requested error bound for a estimator, defined in Problem \ref{prob:sufficient_shots} \\
        $\nu$ & Dummy variable defined as $\qty(\sum_{j=1}^\ell \norm{D_j}) / E_f$ used in Algorithm \ref{algo:error_control} \\
        $T_i$ & The temperature of simulated annealing optimizer at the $i$-th step \\
        $\eta$ & The desired upper bound for the KL divergence \\
        $\alpha_i$, $\beta_i$ & The mixing ratio of values and gradients in the recursive estimator \\
        $d_i$, $d_{\partial_k, i}$ & Drift term, a dummy variable used the proof of Theorem \ref{thrm:sofr_bound} \\
        $A$ & STD bound for the recursive estimator \\
        $B$ & Bias bound for the recursive estimator \\
        $a$, $b$, $E'$ & Dummy variables used in Algorithm \ref{algo:error_control_recursive} \\
        \hline
        $\partial_k$ & Partial derivative with respect to the $k$-th dimension\\
        $\gradient$ & Gradient operator \\
        $\mathrm{Hess}$ & Hessian operator \\
        $\norm{\cdot}_2$ & $L^2$ norm \\
        $D_{KL}$ & Kullback-Leibler divergence \\
        $\Pr$ & Probability distribution \\
        $\hat{\cdot}$ & An estimator \\
    \end{tabular}
    \caption{Table of Notations}
    \label{tab:notations}
\end{table}

\end{document}